\documentclass{article}
\usepackage{amsmath,amsthm,amssymb,enumitem,cleveref,a4wide,xcolor}

\newtheorem{theorem}{Theorem}[section]
\newtheorem{lemma}[theorem]{Lemma} \newtheorem*{lemma*}{Lemma} \newtheorem{proposition}[theorem]{Proposition} \newtheorem{corollary}[theorem]{Corollary} \newtheorem*{claim}{Claim}
\theoremstyle{definition} \newtheorem{definition}[theorem]{Definition} 
\newenvironment{algorithm}{\medskip\begin{center}\begin{tabular}{|p{0.9\textwidth}|}\hline\smallskip} {\\\hline\end{tabular}\end{center}\medskip}

\DeclareMathOperator{\C}{C} \DeclareMathOperator{\CI}{CI}

\newcommand{\Cs}{\C_{\textnormal{sp}}} 
\newcommand{\eps}{\varepsilon} 
\newcommand{\mcY}{\mathcal Y} \newcommand{\mcN}{\mathcal N} 
  
\newcommand{\iobpp}{\mathsf{ioBPP}}

\title{Key-agreement exists if and only if the ``interactive vs non interactive Kolmogorov problem'' is not in ioBPP: a short proof}
\author{Bruno Bauwens and Bruno Loff} 

\begin{document}
\maketitle
  \begin{abstract}
    Ball, Liu, Mazor and Pass \cite{blmp23} proved that the existence of key-agreement protocols is equivalent to the hardness of a certain problem about interactive Kolmogorov complexity. We generalize the statement and give a short proof of the difficult implication. 
  \end{abstract}

\noindent
This note is self-contained, all definitions are given in \Cref{sec:preliminaries}. The definitions are simpler than those of \cite{blmp23}, yet give a stronger result. The proof of the equivalence has two directions. One direction is almost the same as in \cite{blmp23}, so we include it in \Cref{sec:easy-direction}. The other direction is proven in \Cref{sec:hardDirection}.

\section{Preliminaries}\label{sec:preliminaries}

In a key agreement protocol, Alice and Bob exchange messages via a public channel, with the goal of generating a shared secret: a private, random string known only to them. Any third party who sees their communication should be unable to guess the shared string, except perhaps with negligible probability. We may formalize this as follows.

\begin{itemize}[leftmargin=*, label=--]
  \item 
    There are two randomized algorithms A and B, which are given the same input~$n$. 
    They interact by alternatingly sending a single bit. 
    While an algorithm computes its reply, the other is idle. 
    At some point both algorithms terminate. 
    Then, the {\em transcript} $\pi$ of their interaction is the string of exchanged bits. 

  \item 
    Each algorithm also produces a private output string. Let $x$ and $y$ be these outputs. 
    They are the two versions of the secret key, (and hopefully, $x=y$). 

  \item A {\em protocol} is a mapping from $n$ to a random triple $(\pi, x, y)$ generated by some algorithms A and B in this way. 
    The protocol is {\em $t$-time bounded} for $t:\mathbb N \to\mathbb N $ if there exist such A and B with runtime~$O(t(n))$. 
\end{itemize}

\begin{definition}
  Let $\alpha,\eps:\mathbb N \to[0,1]$. We say that a protocol $n \mapsto (\pi, x, y)$: 
  \begin{itemize}
    \item has \emph{agreement $\alpha$} if for large~$n$: $\Pr[x \mathop{=} y] \ge \alpha(n)$,
    \item is \emph{$\eps$-secure} if for every probabilistic polynomial time algorithm $E$ and large $n$:\footnote{
	In some references, like~\cite{Impagliazzo1989limits,blmp23}, algorithm $E$ does not have $n$ as input. 
	This is meaningless because the following protocol would be a key-agreement one. 
	Alice and Bob have no communication, thus $\pi = $``empty string'', and $x=y=n$. 
	Since $E$ has the empty string as input and is time bounded, with probability 1, its output is bounded by a constant, 
	and therefore wrong for large~$n$. 
      }
      \[
	\Pr[E(0^n,\pi) \mathop= x] \le \eps(n).
      \]
  \end{itemize}
  A {\em key-agreement protocol} is a protocol that is $n^{O(1)}$-time bounded, has agreement $1 - \frac{1}{n^c}$ and is $\frac{1}{n^c}$-secure for all~$c$.
\end{definition}

\noindent
A fundamental open problem in cryptography is whether there exists a key-agreement protocol. 
In~\cite{blmp23}, an equivalent problem is given in terms of interactive Kolmogorov complexity.

\bigskip
\noindent
The {\em Kolmogorov complexity} of $x$ given $y$ is $ \C(x \mathop| y) = \min \{|p| : M(p,y) = x\}$, 
where $M$ is a machine that minimizes the function $\C$ up to an additive constant. 

Let $U$ and $V$ be {\em deterministic} Turing machines with a private input, work and output tapes, and a shared communication tape. 
Let $U(a) {\leftrightarrow} V(b)$ denote the triple $(\pi,x,y)$ obtained as follows. 

Initially, both machines have programs $a$ and~$b$ on their respective input tapes. 
Machines $U$ and $V$ alternatingly calculate a single bit and \textit{communicate} it by writing it on the shared communication tape. 
After both machines halt, $\pi$ is the binary string of communicated bits, and $x,y$ are the contents of $U$ and $V$'s private output tapes. 

\begin{definition}
  The $t$-time-bounded {\em interactive Kolmogorov complexity} of a pair $(\pi, x)$ relative to machines $U,V$ is
  \[
    \CI^t_{U,V}(\pi,x) = \min \Big\{|ab| : (\pi,x,x) = U(a) {\leftrightarrow} V(b)\Big\}, 
  \]
  where both $U$ and $V$ terminate after at most $t(|\pi x|)$ steps. 
\end{definition}

\noindent
The pair $U,V$ is {\em optimal} for this complexity if for every other pair $\widetilde U, \widetilde V$ there exists a constant $c$ satisfying
\[
  \forall \pi, x, t :  \CI^{ct\log t}_{U,V}(\pi, x) \le \CI^{t}_{\widetilde U, \widetilde V}(\pi, x) + c.
\]
We fix such $U,V$ and drop the subscripts from $\CI^t$. 
The special case of interactive complexity for $x = $ ``empty string'' and $t = +\infty$ was first studied in~\cite{onlineComplexity} and shown to exceed $\C(\pi) + \Omega(|\pi|)$ for some $\pi$ in~\cite{Bauwens2014asymmetry}.

\begin{definition}\label{def:iobpp}
  A \textit{promise problem} is a pair $(\mcY, \mcN)$ of sets.
  A problem $(\mcY, \mcN)$ is in $\iobpp$ if there exists a probabilistic polynomial time algorithm $A$ and infinitely many $n$ such that 
  \begin{itemize}
    \item for all $z \in \mcY \cap \{0,1\}^{n}$: $\Pr[A(z) = \textnormal{``outside $\mcN$''}] \ge 2/3$, and
    \item for all $z \in \mcN \cap \{0,1\}^{n}$: $\Pr[A(z) = \textnormal{``outside $\mcY$''}] \ge 2/3$. 
  \end{itemize}
\end{definition}

\noindent
For $\mathsf{io}$ complexity classes, the representation is important in ways that are not for normal complexity classes; 
for example, if all representations have odd length then the problem is trivially in~$\iobpp$. 
So, for concreteness, we will represent a pair $(\pi,x)$ of Boolean strings by $0\pi_1 \cdots 0\pi_{|\pi|}1x$, thus $|(\pi,x)| = |\pi\pi x| + 1$.

All results also hold with a more robust variant of~$\iobpp$. 
We say that a problem is in {\em interval-$\iobpp$} if there exists a function $f(n) \le o(n)$ such that the 2 requirements of definition~\ref{def:iobpp} 
hold for infinitely many $n$ and all $z$ with $f(n) \le |z| \le n$. 
Being interval-$\iobpp$ is invariant under linear time transformations of the representation.

\bigskip
\noindent
The following is a simpler version of the main result of~\cite{blmp23}.
\begin{theorem}\label{th:main}
  Let $c>0$, $e>c+3$, $t(n) \in \omega(n \log n)$ with $t(n) \le n^{O(1)}$ and 
  \begin{align*}
    \mcY^t_c &= \{(\pi, x) : \CI^t(\pi, x) \le \C(\pi) + c \log |\pi x|\} \\
    \mcN_e &= \{(\pi, x) : \C(\pi, x) \ge \C(\pi) + e \log |\pi x|\}. 
  \end{align*}
  A key-agreement protocol exists if and only if the promise problem~$(\mcY^t_c, \mcN_e)$ is not in~$\iobpp$. 
\end{theorem}

We make several useful remarks about this theorem in \Cref{sec:remarks}, where we also prove that it implies theorem 3.4 in~\cite{blmp23}.

\medskip
\noindent
{\em More remarks.}
\begin{itemize}[leftmargin=*,nosep,label=--]
  \item 
   The theorem also holds if we replace the set $\mcN_e$ by $\{(\pi,x):\C(x|\pi) \ge e \log |\pi x|\}$. 
 \item 
   It also holds if we use interval-$\iobpp$ instead of~$\iobpp$. 
\item 
  We will prove the theorem for the special case that $c > 3$ and $e > c + 6$. 
  The proof with better parameters is outlined in the footnotes.\footnote{ 
    A variant with prefix-free complexity with $|\pi|$ in the condition holds for all~$0<c<e$. 
    }
\end{itemize}

\medskip
\noindent
The proof that nonexistence of key-agreement implies that $(\mcY^t_c,\mcN_e)$ is in $\iobpp$, is the same as in \cite{blmp23}, so we include it in \Cref{sec:easy-direction}. 
We now prove the other direction.

\section{Algorithms for $(\mcY^t_c, \mcN_e)$ can break key-agreement}\label{sec:hardDirection}

\begin{proposition}\label{prop:hardDirection}
  Let $c > 3$ and $t(n) \in \omega(n\log n)$. 
  If $(\mcY^t_c, \mcN_e)$ is in $\iobpp$, then there exists no key-agreement protocol.  
\end{proposition}

\noindent
Using the Goldreich-Levin theorem it is possible to prove that, if a secure key-agreement protocol exists, i.e. where Eve cannot guess the entire key, then a \emph{strongly secure} key-agreement protocol exists, where Eve cannot even distinguish the agreed-upon key from a uniformly random string.

\begin{definition}
  Let $\eps:\mathbb N \to[0,1]$. We say that a protocol $n \mapsto (\pi, x, y)$ \emph{has leakage $\eps$} if there exists a probabilistic polynomial time algorithm $E$ and infinitely many $n$ such that 
  \[
    \Big| \Pr[E(\pi, x)] - \Pr[E(\pi, U)] \Big| > \eps(n), 
  \]
  where $U$ is a random string in $\{0,1\}^{|x|}$. 
\end{definition}

\noindent
In the above definition, conditional to $|U| = |x|$, string $U$ is independent of the protocol. 
For technical reasons, we consider leakage in protocols that satisfy additional assumptions. 

\begin{definition}
  We say that a protocol $n \mapsto (\pi,x,y)$ is {\em standard} if  $|(\pi,x)| = n$, it has runtime~$O(n)$, 
    it has $(1-1/n^{\omega(1)})$-agreement, and the amount of randomness used by Alice and by Bob only depends on~$n$. 
\end{definition}

\noindent
The following is proven in~\Cref{sec:proof-GL-corollary}.

\begin{corollary}[of \cite{gol-lev:c:one-way,yao:c:oneway}]\label{cor:GLka}
  Suppose that there exist positive $\epsilon,b,d$ such that every standard protocol with $|x| \ge d\log |\pi x|$ and which uses at most $n^{\epsilon}$ bits of randomness, has leakage~$1/n^{b}$. 
  Then key-agreement does not exist. 
\end{corollary}

\noindent
We will prove that an algorithm for the $(\mcY^t_c, \mcN_e)$ problem can be used to break all protocols in the above sense, that 
it can distinguish the agreed key from a random string. 
This is easy to show for a specific type of protocol. 
Recall that a protocol is a randomized algorithm, hence, it induces a map from the chosen randomness by Alice and Bob to transcripts. 
In~\cite{blmp23} a protocol is called Diffie-Hellman-like if this map is bijective.

\begin{definition}
  A protocol is {\em DH-like} if for every $n$ the map $(r_A, r_B) \mapsto \pi$ is bijective. 
\end{definition}

\begin{lemma}\label{lem:DHlikeHaveLeakage}
  Let $c>3$, $d>e$ and $t(n) \ge \omega(n\log n)$. If $(\mcY^t_c, \mcN_e)$ is in $\iobpp$, then every DH-like standard protocol with $|x| \ge d\log |\pi x|$ has leakage~$\Omega(1)$. 
\end{lemma}

\begin{proof} 
  Let $n \mapsto (\pi,x,y)$ be a DH-like standard protocol with $|x| \ge d\log |\pi x|$. 
  We need to show that Eve behaves differently when given the key $x$ vs a random string $U$ of length~$|x|$. 

  Let $D$ be an $\iobpp$-decider for $(\mcY^t_c, \mcN_{e})$ and $n$ be a large length for which the $\iobpp$ requirement holds for all $(\pi,x)$ with $|(\pi,x)|= n$. 
  We prove that
  \[
    \Pr[D(\pi, x) \mathop= \textnormal{``outside $\mcN$''}] \;-\; \Pr[D(\pi, U)\mathop= \textnormal{``outside $\mcN$''}]  \;\ge\; \frac 1 3 - \frac 2 {100}. 
  \]

  \medskip
  \noindent
  {\em Right probability.} Recall that $U$ is independent of $\pi$ conditional to $|U|$. Hence, with probability $1-\tfrac 1 {100}$,
  \[
    \C(\pi, U) - \C(\pi) + O(1) \ge |U| = |x| \ge d \log |\pi x|. 
  \]
  Thus, $(\pi, U) \in \mcN_{e}$ since $d>e$ and $n$ is large. Hence, $D$ answers ``outside $\mcN$'' with probability at most $\tfrac 1 3 + \tfrac 1 {100}$, since $|(\pi,U)| = n$. 

  \medskip
  \noindent
  {\em Left probability.}
  Let $a$ and $b$ be the randomness used by Alice and Bob to obtain $(\pi,x,y)$ on input~$n$. 
  We prove that with probability $1-\tfrac 1 {100}$ and up to $O(1)$ terms: 
  \[
    \CI^t(\pi,x) - (2 + o(1))\log |\pi x| \le |ab| \le \C(ab|n) = \C(\pi|n) \le \C(\pi). 
  \]
  This implies that $(\pi, x)$ is in $\mcY^t_c$, since $c>2$. 
  Thus, $D$ answers ``outside $\mcN$'' with probability almost $\tfrac 2 3$. 

  The left inequality holds with probability $1-1/n^{\omega(1)}$, 
  because programs for $a$ and $b$ in the definition of interactive complexity 
  can be obtained from the randomness of Alice and Bob together with an encoding of~$n$, 
  (thus both to Alice's and Bob's randomness, a description of $n$ in binary is prepended). 
  Note that Bob's key indeed equals Alice's key $x$ with probability $1-1/n^{\omega(1)}$. 
  Also note that the runtime of the protocol is $O(|\pi x|)$, because it is standard. 
  Since $t(n) \ge \omega(n\log n)$, this computation can be simulated on the optimal pair $(U,V)$ in the definition~$\CI$. 

  The 2nd inequality holds with probability $1-\tfrac 1 {200}$, 
  since for fixed $n$, the randomness $a$ and $b$ has a fixed length, by definition of being standard. 
  The equality follows by bijectivity of DH-like protocols. 
  The chain of (in)equalities and the lemma are proven. 
\end{proof}

Already \cite{blmp23} had observed that \cref{prop:hardDirection} can be easily proven for DH-like protocols. 
It is easy to transform any protocol to an approximately DH-like protocol by letting Alice and Bob sending a long enough hash value of their randomness. 
But in general, the breaker only has the transcript of the original protocol. How can he add hash values to simulate the modified protocol? 

The idea is that a short enough hash value looks random (by the left-over-hash lemma). 
Thus, the breaker $E$ may replace hash values by random strings and then call the decider of $(\mcY^t_c, \mcN_e)$. 
The length of the hashes should be just right: short enough to look random, and long enough so that the simulated protocol is approximately DH-like. 
The idea is to guess these lengths. The probability of guessing correctly will be proportional to the leakage.\footnote{ 
  It seems that in \cite{blmp23}, hash functions are used from a set of size poly($n$), which is small. 
  The first author recognizes some similarities to the construction of dispersers in \cite{ats-uma-zuc:j:expanders}, 
  which are also used to obtain hash functions from a set of size~$n^8$ constructed in~\cite{teu:j:shortlists, bmvz:j:shortlist}. 
  In fact, the argument can be finished more briefly using these hashes, but it requires $c > 19$. 
  Below we give an elementary proof, which works with any universal set with $2^{O(n)}$ hash functions. 
  We will use the smaller size $2^{O(n^{2\epsilon})}$ with $\epsilon < 1/2$ to avoid technicalities regarding 
  the runtime of computing hashes and to obtain theorem~\ref{th:main} for any $c>0$, (as explained in the footnotes). 
}

\newcommand{\lemRandomInverseTwoHashes}{
  Consider a partition of a finite set $\mathcal R$ into a family of sets $A_\pi \times B_\pi$ for all $\pi$. 
  For each~$\pi$, let $H_\pi$ and $G_\pi$ be universal sets of hash functions 
  mapping $A_\pi$ to $W_\pi = \{0,1\}^{\lfloor \log |A_\pi|\rfloor}$ and $B_\pi$ to $V_\pi = \{0,1\}^{\lfloor \log |B_\pi|\rfloor}$. 

  Select a random $\pi$ with probability $|A_\pi \times B_\pi|/|\mathcal R|$. Also select random $h \in H_\pi$, $w \in W_\pi$, $g \in G_\pi$, $v \in V_\pi$. 
  With probability $\Omega(1)$, the following inequality is defined and is true:
  \[
    \C (h, h_{A_\pi}^{-1}(w), g, g_{B_\pi}^{-1}(v)) \,\ge\, \log (|H_\pi| \cdot |G_\pi| \cdot |\mathcal R|) - 5. 
  \]
}

\begin{proof}[Proof of \Cref{prop:hardDirection}]
  Let $D$ be a decider for $(\mcY^t_{c}, \mcN_{e})$ that for infinitely many $m$ answers correctly on $m$-bit inputs with probability $1-o(1/m)$. 
  (As usual, by repeated trials we may indeed assume that the success probability is very close to~$1$.) 
  Suppose that this happens for infinitely many {\em even}~$m$. If this is not the case, apply the padding in algorithm $E$ below to length $2n+1$ instead of~$2n$. 

  For $k \le \rho$, let $H_{\rho,k}$ be a universal set of hash functions from $\{0,1\}^\rho$ to $\{0,1\}^k$. 
  Assume that $\log |H_{\rho,k}| \le O(\rho^2)$ and that each $h \in H_{\rho,k}$ is represented as a string of length $\log |H_{\rho,k}| + O(1)$. 
  Say, we hash with random $k \times \rho$ binary matrices, for which $|h| = \rho k \le \rho^2$. 

  Note that in a reduction to the $\iobpp$-decider~$D$, the length needs to be carefully controlled, and this forces a discussion of some details. 
  Our choice of representation of pairs has $|(\pi,z)| = 2|\pi| + 1 + |z|$. 
  For $\ell \ge |(\pi,z)|$, let $(\pi,z)_{:\ell}$ be the {\em padded pair} $(\pi', z')$ such that $|(\pi',z')| = \ell$, string $z'$ extends $z$ by at most 1 bit and $\pi' = \pi 10^k$ for the appropriate~$k$. 

  Let $\pi\langle h,g,w,v\rangle$ be the transcript in which first Alice and Bob exchange~$\pi$. 
  Then Alice sends $h$, bit by bit, while Bob replies with $0$'s. Afterwards Bob sends $g$ in a similar way, and this is repeated with $w$, sent by Alice, and $v$, sent by Bob. 

  Finally, let $\epsilon > 0$ be small enough such that $c > 3 + 4\epsilon$ and $\epsilon < 1/2$. 
  
  \bigskip
  \noindent
  {\em Predicate $E(\pi, z)$.} Let $n = |(\pi,z)|$ and $\rho = \lfloor n^\epsilon \rfloor$. 
  Randomly select $\alpha, \beta \in [\rho]$ (the unknown hash lengths), $h \in H_{\rho,\alpha}$, $g \in H_{\rho,\beta}$ (the hash functions), and $w \in \{0,1\}^{\alpha}$ and $v \in \{0,1\}^{\beta}$ (the presumed hash values of $r_A$ and $r_B$). 
  Output $1$ if
  \[
    D\big((\pi\langle h,g,w,v\rangle, z)_{:2n}\big)
  \] 
  answers ``outside~$\mcN$'', and otherwise, output~$0$. 

  \bigskip
  \noindent
  Note that the padding to length $2n$ is possible because $h,g,w,v$ all have bit size at most $n^{2\epsilon} \le o(n)$. Hence, $|(\pi\langle h, g, w, v\rangle, z)| \le n + o(n) \le 2n$ for large~$n$. 

  Let $d > e + 2\epsilon$. 
  Let $n \mapsto (\pi,x,y)$ be standard with $|x| \ge d\log |\pi x|$ and generated with at most $\rho = \lfloor n^\epsilon \rfloor$ random bits.

  We prove that the protocol has $\Omega(1/n^{2\epsilon})$-leakage. 
  Thus for infinitely many $n$, 
  \begin{equation}
    \Pr[E(\pi, x)] - \Pr[E(\pi, U)] \ge \Omega(1/\rho^2). \label{eq:gl2}
  \end{equation}
  This implies the condition of \Cref{cor:GLka} and hence \Cref{prop:hardDirection}.  
  Let $n$ be large such that $D$ is correct with probaiblity $o(1/n)$ on inputs of length~$n$. 

  \begin{claim}
    $\Pr[E(\pi, U)] \le o(1/\rho^2)$. 
  \end{claim}

  \noindent
  The analysis is similar to before. Let $(\tilde \pi,\tilde U)$ be the input given to the decider~$D$ in algorithm~$E$. 
  Note that $\tilde U$ either equals $U$ or $U$ with a single bit appended. 
  With probability $1 - o(1/\rho^2)$ over the choice of~$U$, 
  \[
    \C(\tilde \pi, \tilde U) - \C(\tilde \pi) \ge |U| - (2\epsilon+o(1))\log n \ge (d-2\epsilon-o(1)) \log |\tilde \pi \tilde U|. 
  \]
  Thus, $(\tilde \pi, \tilde U) \in \mcN_{e}$ for large $n$, by choice of $d > e + 2\epsilon$. 
  Hence, the $\iobpp$ decider outputs ``outside $\mcN$'' with probability $o(1/\rho^2)$. 
  This implies the claim.

  \begin{claim}
    $\Pr[E(\pi, x)] \ge \Omega(1/\rho^2)$. 
  \end{claim}

  \noindent
  It remains to prove this claim, since both claims together, imply inequality~\eqref{eq:gl2}. 

  We view $w$ and $v$ as hashes. For $A \subseteq \{0,1\}^\rho$, fix any pseudo inverse of $h$ in $A$, say the one based on the lexicographic first match in $A$, thus
  \[
    h_A^{-1}(w) = \min \{a \in A : h(a) = w\},
  \]
  when defined. Note that it may be undefined. Similar for $g_A^{-1}$. 
  
  Let $A_\pi \subseteq \{0,1\}^{\rho}$ be the set of Alice's random strings that are compatible with her replies in~$\pi$. 
  Let $B_\pi$ be defined similarly for Bob's randomness. These sets are computable given~$\pi$ and~$n$.

  The input of the decider $D$ in $E$ is 
  \[
    (\tilde \pi, \tilde x) = (\pi\langle h,g,w,v\rangle, x)_{:2n}. 
  \]
  We prove that with probability $\Omega(1/\rho^2)$, the values 
  \[
    a = h^{-1}_{A_\pi}(w) \quad \textnormal{ and } \quad b = g^{-1}_{B_\pi}(v)
  \]
  are defined and satisfy up to $o(\log n)$-terms,
  \begin{equation}\label{eq:chain}
    \CI^t(\tilde \pi, \tilde x) - (2 + 4\epsilon)\log n \le |hagb| \le \C(h,a,g,b) \le \C(\tilde \pi) + \log n. 
  \end{equation}
  Since $c>3+4\epsilon$, this implies $(\tilde \pi, \tilde x) \in \mcY^t_c$ with probability $\Omega(1/\rho^2)$. 
  Hence, for infinitely many~$n$, the decider outputs ``outside $\mcN$'' with probability $\Omega(1/\rho^2)$. This implies the claim.  
  It remains to prove that each of the 3 inequalities from~\eqref{eq:chain} holds with probability $\Omega(1/\rho^2)$ for all~$n$.   

  \begin{itemize}[leftmargin=*]
    \item 
      Let EH be the event that the hash sizes in algorithm $E$ are correct, thus that $\alpha=\lfloor \log |A_\pi|\rfloor$ and $\beta=\lfloor \log |B_\pi| \rfloor$. 
      This happens with probability $1/\rho^2$. 

    \item 
      The following lemma implies that conditional to EH, with probability $\Omega(1)$, the values $a$ and $b$ are defined and satisfy $\C(hagb) \ge |hagb| - O(1)$, 
      i.e., the middle inequality of~\eqref{eq:chain}. 

      \begin{lemma}\label{lem:randomInverseTwoHashes}
	\lemRandomInverseTwoHashes
      \end{lemma}

      The lemma is proven in \Cref{sec:randomInverseHash}. 
      It is applied to $\mathcal R = \{0,1\}^\rho \times \{0,1\}^\rho$ (Alice and Bob's randomness), 
      $H_\pi = H_{\rho,\lfloor \log |A_\pi|\rfloor}$ and $G_\pi = H_{\rho,\lfloor \log |B_\pi|\rfloor}$. 
      Event EH implies that $W_\pi = \{0,1\}^\alpha$ and $V_\pi = \{0,1\}^\beta$. 
      Thus, the distribution of $\pi,h,w,g,v$ in \Cref{lem:randomInverseTwoHashes} is the same as in $E$ conditional to event~EH. 
      The middle inequality of~\eqref{eq:chain} is proven.

    \item 
      Given $\tilde \pi,n$ one can compute $A_\pi$ and $B_\pi$.
      The right inequality of~\eqref{eq:chain} follows by the choice of $a$ and~$b$, assuming that they are defined. 

    \item
      Let EA be the event that $x = y$. 
      By agreement, EA happens with probability~$1-1/n^{\omega(1)}$. 

    \item 
      We prove the left inequality of \eqref{eq:chain} conditional to event EA.  

      Alice knows $n,a,\alpha,h$, and Bob~$n, b,\beta,g$. 
      Together they possess $|ahbg|+ (2+4\cdot \epsilon + o(1))\log n$ bits of information.  
      They generate the transcript $\pi\langle h,g,w,v\rangle$ in the logical way. 
      First they simulate the protocol and obtain~$\pi$. 
      Then, they send $h,g,v,w$ in turns, where $w$ is the $\alpha$-bit hash of $a$ and similar for~$v$. 
      Using $n,\pi$ and their randomness, they each obtain~$x$, (assuming EA). It remains to add the padding to obtain $\tilde \pi, \tilde x$. 

      Recall that $\epsilon < 1/2$. 
      The above procedure runs in time $O(n)$, for example for hashing with random matrices using
      $
	H_{\rho,k} = \{a \mapsto Ma : M \in \{0,1\}^{k \times \rho} \}. 
	$
      Note that evaluating $Ma$ takes time $O(k\rho) \le O(n^{2\epsilon}) \le o(n)$. 
      The optimal machines $(U,V)$ in the definition of $\CI$ simulate everything in time $O(n \log n)$. 
      Since $t(n) \ge \omega(n \log n)$, the inequality holds. 
  \end{itemize}
  In summary, conditional to events EA and EH, all requirements are satisfied with probability~$\Omega(1)$. 
  Since EA and EH happen simultaneously with probability $\Omega(1/\rho^2)$, this implies the chain~\eqref{eq:chain} 
  and the claim. 
  Proposition~\ref{prop:hardDirection} is proven.\footnote{
    Recall that proposition~\ref{prop:hardDirection} assumes $c>3$.  
    To prove it for $c>0$, note that the excess term in the left and right inequality of~\eqref{eq:chain} is needed for providing~$n$. 
    One might think that $n$ should be included in the transcript, so that it also increases the right side of the inequality. 
    This solves the problem for the right inequality. 
    But this is not enough for the left one, because of the time bound in $\CI^t$ and $n$ might be computationally deep. 
    The solution is to make the program length reveal~$n$. 

    In the construction of~$E$, we prepend a random string $s$ of length $n/2$ to the transcript $\tilde \pi$. Only the length of $s$ will be used. 
    Let $\bar a$ be a prefix code of length $|a| + 2\log |a| + 1$. 
    Now Alice's program is $\bar a  \bar \alpha \bar h  s$, which has length $|ahs| + O(\epsilon \log n)$. Note that $s$ is not given in a prefix way. 
    In the beginning, Alice strips the 3 prefix codes from the program to obtain~$s$. She sends $s$ to Bob. 
    They both use $n = 2|s|$ or $n=2|s|+1$, depending on whether there are infinitely many even or odd $n$ for which the decider works. 
    Then the protocol and the analysis proceeds as before. This removes the $2\log n$ and $\log n$ terms from the left and the right inequalities in~\eqref{eq:chain}.  
  }
\end{proof}

\bibliographystyle{alpha}
\bibliography{theory-3}

\appendix

\section{Remarks about the $(\mcY^t_c,\mcN_e)$ problem}\label{sec:remarks}

Recall that
  \begin{align*}
    \mcY^t_c &= \{(\pi, x) : \CI^t(\pi, x) \le \C(\pi) + c \log |\pi x|\} \\
    \mcN_e &= \{(\pi, x) : \C(\pi, x) \ge \C(\pi) + e \log |\pi x|\}. 
  \end{align*}

\begin{itemize}[leftmargin=*,label=--]
  \item The problem $(\mcY^t_c,\mcN_e)$ depends on 3 parameters: $c,e,t$. 
    Decreasing $c$ or $t$, or increasing $e$, can only make the problem easier to solve (because this makes the sets smaller). 
     \Cref{th:main} implies that for all $c,e,t$ satisfying the conditions, either all problems $(\mcY^t_c,\mcN_e)$ are in $\iobpp$ or none are. 
  \item 
    If $r> 1$, then the set $\mcY^t_c \cap \mcN_{c+r}$ is finite, because 
    \[
      \CI^t(\pi, x) \ge \C(\pi, x) - (1 + o(1))\log |\pi|
    \]
    holds by concatenating programs and adding a prefix free description of the splitting point. 
    Thus for $e > c+1$, the promise problem cannot be trivially excluded from~$\iobpp$. 
\end{itemize}

\begin{lemma}\label{lem:}
  If $r > 3$, then the problem $(\mcY^t_c, \mcN_{c+r})$ can be solved in space $O(t \log t)$. 
\end{lemma}

\begin{proof}
  For an integer $s$, 
  let $\Cs^s(\cdots)$ be the Kolmogorov complexity on a reference machine with programs that use $s$ cells on a binary work tape. 
  Symmetry of information holds with a space bound:
  \[
    \Cs^{s'}(\pi) + \Cs^{s'}(x|\pi) \le \Cs^s(\pi,x) + (1+o(1))(\log \Cs^s(\pi) + \log \C^s(x|\pi)), 
  \]
  where $s' = 3s + 10|\pi|$. This follows from the standard argument of symmetry of information and careful handling of the logarithmic terms. 
  (See \cite[theorem 3.13, p35]{Longpre1986resource} or~\cite[theorem 1]{Bauwens2022Inequalities} for recent improvements on $s'$. 
  Note that some parameters are almost equal to the program lengths, see~\cite[proposition 2]{BauwensPlainAdditivity}). 

  Let $\eps > 0$ be small such that $2\eps < r - 3$. Let $\tilde c$ be a constant that we choose later. 

  \medskip
  \noindent
  {\em Algorithm on input $(\pi,x)$.} 
  Let $n = |\pi x|$ and $s' = 3\tilde c \cdot t(n) \log t(n)+10n$. If 
  \begin{equation}\label{eq:inequalityDecidability}
    \Cs^{s'}(x|\pi) \le (c+2+\eps) \log |\pi|,
  \end{equation}
  then answer ``outside $\mcN$'', otherwise, answer ``outside $\mcY$''. 

  \medskip
  \noindent
  We verify correctness. If the algorithm outputs ``outside $\mcN$'', then inequality \eqref{eq:inequalityDecidability} holds. 
  Thus, $\C(\pi, x) < C(\pi) + C(x|\pi) + (1+\eps)\log n \le \C(\pi) + (c+r) \log n$ holds, and indeed, $(\pi,x)$ is outside~$\mcN_e$. 

  Let $d = \Cs^{s'}(x|\pi) - (c+2+\eps) \log n$. 
  If the algorithm outputs ``outside $\mcY$'', then $d > 0$. 
  Space-bounded symmetry of information implies
  \begin{align*}
    (c+2+\eps) \log \pi + \C(\pi) & = \Cs^{s'}(x|\pi) + \C(\pi) - d \\
     & \le \Cs^{s}(\pi, x) + (1+o(1))(\log n + \log d) - d\\
     &\le \CI^t(\pi, x) + (2+o(1))\log n,
  \end{align*}
  where $s = \tilde c \cdot t(n) \log t(n)$ and $\tilde c$ is chosen large enough for the simulation of the optimal pair $(U,V)$ in $\CI$. 
  This implies $(\pi, x) \not\in \mcY^t_c$. 
  Thus also if the algorithm outputs ``outside $\mcY$'', it is correct. 
\end{proof}

\bigskip
\noindent
The main result of \cite{blmp23} considers a similar problem on triples $(\pi, x, y)$. We explain that it is a corollary of theorem~\ref{th:main}.
For $d>0$, let
\begin{align*}
  \CI^t(\pi,x,y) &= \min \{|ab| : (\pi,x,y) = U(a){\leftrightarrow} V(b) \textnormal{ in time } t(|\pi|)\} \\
  Q^t_d &= \{(\pi,x,y) : |\pi|=|x|=|y| \textnormal{ and } \max(\C^t(x|y), \C^t(y|x)) \le d\log |\pi|\}, 
\end{align*}
and $Q_0 = \{(\pi,x,x) : |\pi|=|x|\}$. 

\begin{theorem}[\cite{blmp23}]\label{th:blmp23}
  Let $d\ge 0$, $c> 3$, $e-c > 9 + 2d$ and $t(n) = n^\gamma$ for some $\gamma>1$. Key-agreement exists if and only if the promise problem 
  \begin{align*}
     \widetilde{\mcY}^t_{c,d} &= Q_d \cap \{(\pi,x,y) : \CI^t(\pi,x,y) \le \C(\pi) + c\log |\pi|\}\\
     \widetilde{\mcN}_{e,d} &= Q_d \cap \{(\pi,x,y) : \C(\pi,x,y) \ge \C(\pi) + e\log |\pi|\}
  \end{align*}
  is not in $\iobpp$.
\end{theorem}    

\begin{proof}[Proof using theorem~\ref{th:main}.]
  Assume $0 < \eps  < c-3$. If $(\pi,x) \in \mcY^t_{c}$ and $|\pi|=|x|$, then $(\pi,x,x) \in \widetilde \mcY^t_{c+\eps,d}$. 
  Similarly, if $(\pi,x) \in \mcN_{e}$, then $(\pi,x,x) \in \widetilde \mcN_{e-\eps,d}$.

  Thus, if $(\widetilde \mcY^t_{c,d}, \widetilde \mcN_{e,d}) \in \iobpp$, 
  then the variant of the $(\mcY^t_{c+\eps}, \mcN_{e-\eps})$ with $|\pi|=|x|$ is in~$\iobpp$. 
  But then for $s = n^{\gamma - \eps}$, also $(\mcY^s_{c+2\eps}, \mcN_{e-2\eps}$ is in~$\iobpp$, because we can padd the shorter string of the pair and this affects the complexities by~$O(1)$.  
  By  theorem~\ref{th:main}, key-agreement does not exist.  

  \bigskip
  \noindent
  Now the other direction is proven. 

  \begin{claim}
    If $s \ge \omega(t \log t)$ and
    $(\mcY^{s}_{c+d+\eps}, \mcN_{e-d-\eps}) \in \iobpp$, then $(\widetilde \mcY^t_{c,d}, \widetilde \mcN^t_{e,d}) \in \iobpp$. 
  \end{claim}

  \begin{proof}
    This follows from 
    \begin{align*}
      (\pi,x,y) \in \widetilde \mcY^t_{c,d}  \quad &\Longrightarrow \quad (\pi,x) \in \mcY^s_{c+d+\eps} \\
      (\pi,x,y) \in \widetilde \mcN^t_{e,d}  \quad &\Longrightarrow \quad (\pi,x) \in \mcN_{e-d-\eps}. 
    \end{align*}
    For the first implication, observe that
    \[
      \CI^s(\pi,x) \le \CI^t(\pi,x,y) + (1+o(1))\C^t(y|x). 
    \]
    Indeed, ``Bob's'' program $b$ has output $y=V(b,\pi)$, so we prepend a program to $b$ that maps $y$ to~$x$. 
    By definition of $Q^t$, we have $\C^t(y|x) \le d\log |\pi|$. 
    We need a large time bound $s$ to account for the simulation of the programs on the fixed optimal machines in~$\CI$. 
    The first implication is proven. 

    For the second, note that dropping $y$ in the complexity, can only decrease the complexity by $\C(y|\pi, x) + O(\log \C(y|\pi,x))$. 
    By definition of $Q_d$, this is at most $(d+o(1))\log |\pi|$. The claim is proven. 
  \end{proof}
  
  \noindent
  Now we finish the proof of~\Cref{th:blmp23}. 
  Assume key-agreement does not exist. 
  By \Cref{th:main} this implies $(\mcY^{n^{\gamma}}_{c+d+\eps}, \mcN_{e-d-\eps}) \in \iobpp$ for every $\gamma > 1$. 
  By the claim, this implies $(\widetilde \mcY^{n^\gamma}_{c,d}, \widetilde \mcN^{n^\gamma}_{e,d}) \in \iobpp$ for every $\gamma > 1$. 
\end{proof}

\section{Proof of \Cref{cor:GLka}}\label{sec:proof-GL-corollary}

\begin{lemma}[\cite{gol-lev:c:one-way,yao:c:oneway}]\label{lem:GL}
There exists an oracle-aided probabilistic polynomial time algorithm $A$ and a positive polynomial $p$ 
such that the following holds for all integers $\ell, n$, all $\alpha > 0$, all distributions $Q$ over $\{0,1\}^* \mathop\times \{0,1\}^n$, 
and every polynomial time algorithm~$E$. If
\[
\Big| 
       \Pr \big[E(\pi,r_1 \cdots r_\ell, x \odot r_1 \cdots x \odot r_\ell)\big] 
     - \Pr \big[E(\pi,r_1 \cdots r_\ell, U)\big] 
     \Big|  \;\ge\; \alpha
\]
for $(\pi,x) \sim Q$ and for uniformly random $r_i \in \{0,1\}^n$, $U \in \{0,1\}^\ell$, then
\[
\Pr_{(\pi,x) \sim Q} [A^E(0^n, 0^\ell, 0^{\left\lceil \!\tfrac 1 \alpha\! \right\rceil}, \pi) = x] \;\ge\; p(\alpha, 2^{-\ell}, \tfrac 1 n). 
\]
\end{lemma}

\begin{proof}[Proof Sketch]
  Using Yao's next-bit predictor theorem (\S9.3.1 in~\cite{aro-bar:b:complexity}), we can obtain from $E$ a randomized predictor $P(\pi, r_1, \ldots, r_i, r_1 \odot x, \ldots, r_{i-1} \odot x)$, such that, for a typical choice of $r_1, \ldots, r_{i-1}$, when given $x \odot r_1, \ldots, x \odot r_{i-1}$ as advice the predictor outputs $x \odot r_i$ with probability $\frac{1}{2} + \frac{\alpha}{\ell}$ over the choice of~$r_i$. We are not assuming that $A$ has access to $x \odot r_1, \ldots, x \odot r_{i-1}$ for such a typical choice, so we have $A$ sample $r_1, \ldots, r_{i-1}$ at random and guess the values of $x \odot r_1, \ldots, x \odot r_{i-1}$, and with probability $\Omega(2^{-\ell})$ the $r$'s are typical and the guess is correct with probability $\alpha/\ell$.

    So, conditioned on this event, we have obtained a predictor which allows us to output pairs $(r, x \odot r)$, for random $r$, which are correct with probability $\ge \frac{1}{2} + \alpha/\ell$. The Goldreich--Levin theorem (\S9.3.2 in~\cite{aro-bar:b:complexity}) then says that from such a predictor we can obtain a list of size polynomial in $\ell/\alpha$, one of whose elements is $x$. Then $A$ outputs a random element in this list. 
\end{proof}

\begin{proof}[Proof of  \Cref{cor:GLka}.]
  Note that the definition of key-agreement protocol does not change if we rescale $n$ polynomially: 
  if for $\delta>0$, on input $n$, we return the triple $(\pi,x,y)$ from the input $n^{\delta}$ instead, then the protocol remains key-agreement. 
  Indeed, if $\eps(n) < n^{-c}$ for all $c$, then also $\eps(n^{\delta}) < n^{-c}$ for all $c$. Similar for the neglibility of agreement failure. 
  Hence, it is enough to prove that under the assumption of the corollary, there exist no key-agreement protocols with runtime at most $n^{\epsilon/2}$. 
  For such protocols, the lengths $|x|,|y|,|\pi|$ are at most~$n^{\epsilon/2}$. 

  Fix such a protocol $n \mapsto (\pi,x,y)$ with agreement $1-1/n^{\omega(1)}$. 
  Consider a new protocol $n \mapsto (\pi',x',y')$ where Alice and Bob first obtain $(\pi,x,y)$ on input~$n$. 
  Then Alice sends $d\log n$ strings of length $|x|$ to Bob, so that $\pi' = \pi r_1 \ldots r_\ell$. 
  Alice's key is $x' = x \odot r_1 \dots x \odot r_\ell$ and Bob's key is similar. 
  This protocol has the same agreement $1 - 1/n^{\omega(1)}$. 
  To make the protocol standard, we apply the right amount of padding to $\pi'$ so that $|(\pi',x')| = n$.  
  Finally, we ensure that some amount of randomness is read that only depends on $n$. Now, the new protocol with mapping $n \mapsto (\pi',x',y')$ is standard. 

  By assumption of \Cref{cor:GLka}, this protocol has leakage $n^{-b}$. 
  Thus, there exists an adversary $E$ that distinguishes $x'$ from random with probability-difference $n^{-b}$. 
  By lemma~\ref{lem:GL}, this gives us an adversary $A$ that predicts $x$ with non-negligible probability.
  Hence, our original protocol $n \mapsto (\pi,x,y)$ is not a key-agreement one. 
  We conclude that there exist no key-agreement protocols with runtime $n^{\epsilon/2}$, and hence, no key-agreement protocols in general. 
\end{proof}

\section{Proof of \Cref{lem:randomInverseTwoHashes}} \label{sec:randomInverseHash}

\begin{definition}\label{def:}
  We say that a variable $V$ is {\em $\gamma$-close to uniform} in $R$ if its statistical distance to $U_R$ is at most~$\gamma$, 
  where $U_R$ is a uniformaly distributed variable in~$R$. 
\end{definition}

\noindent
The proof uses the following observations. 

\begin{enumerate}[leftmargin=*, label=(\alph*)]
  \item Let $V$ be $(\gamma-\delta)$-close to uniform in $R$. 
    If $S \subseteq R$ and $|S| \ge \gamma|R|$, then $\Pr[V {\in} S] \ge \delta$. 

   \item \label{it:KolmCloseToUnifrom} If $V$ is $(1-2\eps)$-close to uniform in $R \subseteq \{0,1\}^*$, then $\Pr[\C(V) \ge \log (\eps|R|)] \ge \eps$. 

   \item \label{it:uniformInSubset} If $V$ is uniform in $S \subseteq R$ and $|S| \ge \eps |R|$, then $V$ is $(1-\eps)$-close to uniform in~$R$. 

   \item \label{it:mixtureClose} Let sets $R_1, R_2, \ldots$ partition a finite set $\mathcal R$. 
     For each~$\pi$, let $V_\pi$ be a random variable that is $\gamma$-close to uniform in~$R_\pi$. 
     Let $M = V_\pi$ where $\pi$ randomly selected with probability $|R_\pi|/|\mathcal R|$. 
     Then, $M$ is $\gamma$-close to uniform in~$\mathcal R$. 

   \item \label{it:hash}
     Let $H$ be a universal set of hash functions from $A$ to $W$ with $|W| \le |A|$. 
     For a random $h \in H$ and $a \in A$, the pair $(h(a),h) \in W \times H$ is supported on a set of size at least $\tfrac 1 2 |H| \cdot |W|$. 

   \item \label{it:hashInv}
     Let $H,A,W$ be as in~\ref{it:hash}. 
     For a random $h \in H$ and $w \in W$, the pair $(h,h^{-1}_A(w)) \in H \times A$ is defined with probability $1/2$ and 
     conditional to being defined, it is uniformly distributed on a set of size $\tfrac 1 2 |H| \cdot |W|$. 
\end{enumerate}

\begin{proof}
  (a) The statistical distance between two measures is the maximal difference of the probabilities of some set. Hence, $\Pr[V {\in} S] \ge \Pr[U {\in} S] - \gamma + \delta \ge \delta$. 

  \smallskip
  \noindent
  (b) Set $\gamma = 1-\eps$ and $\delta = \eps$ in (a). Less than $N$ strings $x$ satisfy $\C(x) < \log N$. 

  \smallskip
  \noindent
  (c) $\Pr[V {\in} S'] - \Pr[U_R {\in} S']$ is maximized for $S' = S$, and this maximum equals $1-|R|/|S|$. 

  \smallskip
  \noindent
  (d) 
  A uniform random variable in $\mathcal R$ is obtained as $U_{R_\pi}$, where $\pi$ is selected in the same way. 
  The statistical distance is now the weighted average of the distances between $U_{R_\pi}$ and $V_\pi$. 

  \smallskip
  \noindent
  (e) This is a special case of the left-over-hash lemma, see e.g.~\cite[lemma 21.25 p445]{aro-bar:b:complexity}. We prove it for convenience. 

  If $p$ has $d$ nonzero entries, then $||p||_1 \le \sqrt d ||p||_2$, (by the Cauchy-Schwartz inequality). 
  Let $p$ be the vector whose entries are the probabilities of $(h(a),h)$ in all elements in $W \times H$. Thus $||p||_1 = 1$. 

  Note that $||p||_2^2$ is the probability that $(h(a),h) = (h'(a'),h')$ for an independent copy $(h',a')$ of $(h,a)$. 
  If $a \not= a'$, then $\Pr[h(a) \mathop= h(a')] = 1/|W|$ by universality of~$H$. Hence, 
  \[
     \frac 1 d = \frac 1 d ||p||_1^2 \le ||p||_2^2 \le \frac 1 {|H|} \big( \frac 1 {|A|} + \frac 1 {|W|} \big) \le \frac 2 {|H| \cdot |W|}.
  \]

  \smallskip
  \noindent
  (f) For each pair $(w,h)$ on which $p$ in the above proof is positive, the pair $(h,h^{-1}_A(w))$ is defined. 
  Moreover, this mapping is injective, and injective mappings preserve uniformity of random variables on their support. 
\end{proof}


\begin{proof}[Proof of \Cref{lem:randomInverseTwoHashes}.] 
  Fix $\pi$ and for random $h,w,g,v$ let
  \[
    X_\pi = (h, h^{-1}_{A_\pi}(w)) \qquad \textnormal{and} \qquad Y_\pi = (g, g^{-1}_{B_\pi}(v)). 
  \]
  By \ref{it:hashInv}, $X_\pi$ is defined with probability $1/2$ and uniform on a set of size $\tfrac 1 2 |H_\pi| \cdot |W_\pi| \ge \tfrac 1 4 |H_\pi| \cdot |A_\pi|$ when defined. 
  Similar for $Y_\pi$.  
  Hence, $(X_\pi, Y_\pi)$ is defined with probability $1/4$ and when defined, it is uniform on fraction $1/16$ of the elements in 
  $H_\pi \times A_\pi \times G_\pi \times B_\pi$. 
  Thus, $(X_\pi, Y_\pi)$ is $\tfrac {15} {16}$-close to uniform in this set, by \ref{it:uniformInSubset}. 

  Let $\widetilde X = X_\pi$ and $\widetilde Y = Y_\pi$ for $\pi$ randomly selected as in the lemma. 
  By \ref{it:mixtureClose}, $V = (\widetilde X, \widetilde Y)$ is $\tfrac {15}{16}$-close to uniform in the set
  \[
     \bigcup_\pi H_\pi \times A_\pi \times G_\pi \times B_\pi.
     \vspace{-2mm}
  \]
  The lemma follows by applying \ref{it:KolmCloseToUnifrom} with $\eps = 2^{-5}$ conditioned to $V$ being defined.  
\end{proof}

\section{Breaking key-agreement means solving $(\mcY^t_c,\mcN_e)$}\label{sec:easy-direction}

\begin{proposition}\label{prop:easyDirection}
  Let $\gamma > 6$ and $t(n) \le n^{O(1)}$. 
  If key-agreement does not exist, then the promise problem $(\mcY^t_{c},\mcN_{c+\gamma})$ is in~$\iobpp$. 
\end{proposition}

\begin{proof}
  It suffices to assume that a single specific protocol is not a key-agreement one. 
  The proof assumes that Eve can break this protocol, not just with non-negligible probability, but with probability close to~$1$. 
  This assumption is without loss of generality because of the following. 

  \begin{theorem}\label{th:Holenstein}
    Let $\alpha > \beta$. If key-agreement does not exist, then each protocol with agreement $1 - \tfrac 1 {n^{\alpha}}$ is not $(1 - \tfrac 1 {n^{\beta}})$-secure. 
  \end{theorem}

  \noindent
  A $1$-bit version of this theorem was proven in~\cite{Holenstein2005stoc,Holenstein2006phd}, 
  and the general version follows from a weak variant of the Goldreich--Levin theorem, see~\cite[appendix B.2]{blmp23}.

  Let $T(n)$ be a polynomial upper bound of $t(n)$. 
  Consider the following protocol. 

  \begin{algorithm}
    {\em Levin-search protocol with parameter~$c$.} Input: an integer $n$. 
      \begin{itemize}[leftmargin=2em,label=--]
      \item 
	Alice randomly samples $\ell_A \in [2n]$ and $a \in \{0,1\}^{\ell_A}$.
	\newline Bob samples $\ell_B \in [2n]$ and $b \in \{0,1\}^{\ell_B}$. 

      \item 
	They view $a$ and $b$ as programs and evaluate their interaction $(\pi,x,y) = U(a) {\leftrightarrow} U(b)$. 
	If one of the program's runtime exceeds $T(n)$, it always replies with $0$ (the value does not matter), until it has replied $T(n)$ times. 

      \item 
	Alice and Bob check whether $x=y$ using hashes of size $\lceil (c+5)\log n\rceil$ that fail with probability at most $n^{-c-5}$. 
	Thus, Alice sends a hash function $h$ and value~$h(x)$, then Bob replies whether this equals $h(y)$ by sending 1 or~0. 

      \item 
	If equal, Alice outputs $x$ and Bob $y$. Otherwise, they output the empty string. 
    \end{itemize}
  \end{algorithm}

  \medskip
  \noindent
  Note that this protocol has agreement $1-n^{-c-5}$, 
  because if $x \neq y$, then both players output the empty string, unless there is a hash collision. 
  Also, observe that the above protocol runs in time $O(T(n))$, which is polynomial by definition of~$T$. 

  If key-agreement does not exist, then Holenstein's theorem implies the existence of an adversary $E$ which guesses the output with high probability. 
  We prove that the following is an $\iobpp$~algorithm.

  \begin{algorithm}
    {\em Algorithm for solving $(\mcY^t_c, \mcN_{c+\gamma})$ on input $(\pi, x)$.}  
  \begin{itemize}[leftmargin=*, label=--]
    \item 
      Let $n = |(\pi,x)|$. 

      Let $h$ be a random hash function with $|h(x)| = (c+5)\log n$. 

      Let $\tilde \pi$ be the transcript obtained from $\pi$ after which Alice sends $h, h(x)$ and Bob replies with $1$. 

    \item 
      If $E(n,\tilde \pi) = x$, output ``outside $\mcN$''. Otherwise, output ``outside $\mcY$''. 
  \end{itemize}
  \end{algorithm}
  
  \noindent
  Let us prove correctness for all $(\pi, x)$ in $\mcN_{c+\gamma}$, thus if $\C(\pi, x) \ge \C(\pi) + (c+\gamma) \log |\pi x|$. 
  Note that $|\pi x| \le n \le 2|\pi x|$. 
  For large~$|\pi|$, symmetry of information implies 
  \[
    \C(x | \pi) \ge \C(\pi, x) - \C(\pi) - (1+o(1))\log |\pi x| \ge (c+\gamma-1-o(1)) \log n, 
  \]
  Thus with probability~$5/6$, 
  \[
    \C(x | \tilde \pi) \ge (\gamma-1-5-o(1)) \log n, 
  \]
  because adding a random hash function $h$ to the condition does not change Kolmogorov complexity, 
  and adding the hash value $h(x)$, decreases the complexity by at most $|h(x)| = (c+5)\log n + O(1)$.\footnote{
    With a direct argument that does not use symmetry of information, 
    it is enough to require that $(c + \gamma - o(1))\log n > |h(x)|$.
  }
  Recall that $\gamma > 6$. 
  Thus no probabilistic program can guess $x$ with probability $4/5$ if it only knows $\tilde \pi$, (formally, one could apply the coding theorem). 
  Hence, with probability at least $2/3$, the guess of $E$ is wrong, i.e. $E(n, \tilde\pi) \neq x$, and the correct answer ``outside $\mcY$'' is given. 

  \bigskip
  \noindent
  Now we prove correctness for infinitely many $n$ and inputs $(\pi, x) \in \mcY^t_{c}$ with $|(\pi,x)|=n$. 
  Recall that the Levin-search protocol has agreement $1-n^{-c-5}$. 
  Thus by~\Cref{th:Holenstein}, for every $\beta < c+5$ there exists an Eve $E$ and infinitely many~$n$ such that Eve guesses the key with probability $1-n^{-\beta}$. 
  We will later see what the right value of $\beta$ should be (it will be any constant $> c+4$). 
  Fix a large such~$n$. 

  To show correctness, we prove the contrapositive of the requirement for $(\pi,x) \in \mcY^t_c$ with $|(\pi,x)|=n$: 
  {\em if the answer is ``outside $\mcY$'' with probability $> 1/3$, then $(\pi,x)$ is outside~$\mcY^t_c$.} 

  Consider sets $S_{n,\ell}$ of pairs for which Eve guesses the key wrong with probability more than~$1/3$: 
  \[
    (\pi, x) \in S_{n,\ell} \quad \Longleftrightarrow \quad \CI^t(\pi, x) = \ell, \quad |(\pi,x)|=n  \quad \textnormal{and} \quad \Pr[E(n, \tilde \pi) \mathop {\not=} x\mid \pi,x] > \tfrac 1 3. 
  \]
  Note that this set is computable on input~$\ell$ and $n$. 
  We need to check that for infinitely many $n$ and for all $\ell$, that the set $S_{n,\ell}$ is disjoint from $\mcY^t(\pi,x)$. 

  For large $n$, we only need to consider $\ell \le 2n$ because $(\pi,x) \in \mcY^t_c$ implies 
  \[
    \CI^t(\pi,x) \le \C(\pi) + c\log |\pi x| \le 2|(\pi,x)|. 
  \]

  What is the probability that on input $n$, the Levin-search protocol produces a fixed $(\pi, x) \in S_{n,\ell}$ satisfying $(\pi,x,x) = U(a) {\leftrightarrow} U(b)$?
  By definition of interactive complexity, this is at least $\tfrac 1 {(2n)^2} 2^{-\ell}$: 
  Alice's and Bob's program lengths $\ell_A$ and $\ell_B = \ell-\ell_A$ are each guessed with probability $1/(2n)$, 
  and then the programs $a,b$ in the definition of $\CI$ are guessed with probability $2^{-\ell}$. 

  Thus each pair in $S_{n,\ell}$ contributes at least $\tfrac 1 3 \cdot \tfrac 1 {(2n)^2} 2^{-\ell}$ to the failure probability of Eve. 
  Since the total failure probability is at most $n^{-\beta}$, we conclude that
  \[
    |S_{n,\ell}| \le O(n^{2-\beta} \cdot 2^\ell).
  \]
  Since $S_{n,\ell}$ is computable, each $(\pi,x) \in S_{n,\ell}$ with $\ell \le 2n$ satisfies up to $o(\log n)$ terms,
  \begin{align*}
    \C(\pi,x) & \le \log |S_{n,\ell}| + 2\log n\\
          & \le \ell + (2-\beta)\log n + 2\log n \\
	  & = \CI^t(\pi,x) + (4-\beta)\log n, 
  \end{align*}
  where the last equality holds by definition of $S_{n,\ell}$. 
  Thus, 
  \[
    \CI^t(\pi,x) \ge \C(\pi) + (\beta - 4 + o(1))\log |\pi x|. 
  \]
  Hence, for large $n$, the set $\mcY^t_c$ is disjoint from $S_{n,\ell}$, provided we choose $\beta > c + 4$, and this is enough for Holenstein's theorem with $\alpha = c+5$. 

  We proved that the answers are correct for infinitely many input sizes. 
  Hence, the promise problem is in $\iobpp$ and the proposition is proven.\footnote{
    Proposition~\ref{prop:easyDirection} holds for all $\gamma > 3$.  
    This follows from the next improvements. 
    Firstly, we have $\C(\pi,x) \le \CI^t(\pi,x) + (3-\beta + o(1))\log n$, because the parameter $\ell$ can be extracted from the index length of $S_{n,\ell}$ with $O(\log \log n)$ information. 
    This means that, we may choose $\alpha$ and $\beta$ to be slightly above $3$. 
    In the Levin-search protocol, we use a hash of size $(3+c+\eps)\log |\pi x|$, with $\eps$ arbitrarily close to $0$. 
    As explained in a previous footnote, we may choose $\gamma$ slightly above $(3+\eps)$, thus also arbitrarily close to~$3$. 
  }$^,$\footnote{
    To prove that the problem is in interval-$\iobpp$, we call $E(n',\tilde \pi)$ 
    with increasing values of $n'$ starting from $n = |(\pi,x)|$ up to $n \log n$. 
    This affects the bound on $|S_{n',\ell}|$ by a factor $\log^2 n$ and the bound $\C(\pi,x) \le \log |S_{n', \ell}| + \log n'$ 
    by an additive term $O(\log \log n)$. Both shifts can be commpensated by arbitrarily small shifts of~$\gamma$. 
  }
\end{proof}

\end{document}